\documentclass[12pt,oneside]{amsart}  
\usepackage{amssymb}
\usepackage{amsfonts}
\usepackage{amsthm}
\usepackage{amsthm}
\usepackage{amscd}
\usepackage{wrapfig}
\usepackage[matrix, arrow, curve]{xy}
\usepackage{graphicx}

\numberwithin{equation}{section}
\theoremstyle{plain}
\newtheorem{theorem}{Theorem}[section]

\newtheorem{corollary}[theorem]{Corollary}
\newtheorem{lemma}[theorem]{Lemma}
\theoremstyle{definition}
\newtheorem{definition}[theorem]{Definition}
\theoremstyle{remark}
\newtheorem{remark}{Remark}[section]
\newtheorem{example}[theorem]{Example}

\begin{document}
\newcommand{\M}{\mathcal{M}}
\newcommand{\F}{\mathcal{F}}

\newcommand{\Teich}{\mathcal{T}_{g,N+1}^{(1)}}
\newcommand{\T}{\mathrm{T}}
\newcommand{\corr}{\bf}
\newcommand{\vac}{|0\rangle}
\newcommand{\Ga}{\Gamma}
\newcommand{\new}{\bf}
\newcommand{\define}{\def}
\newcommand{\redefine}{\def}
\newcommand{\Cal}[1]{\mathcal{#1}}
\renewcommand{\frak}[1]{\mathfrak{{#1}}}
\newcommand{\Hom}{\rm{Hom}\,}
\newcommand{\refE}[1]{(\ref{E:#1})}
\newcommand{\refCh}[1]{Chapter~\ref{Ch:#1}}
\newcommand{\refS}[1]{Section~\ref{S:#1}}
\newcommand{\refSS}[1]{Section~\ref{SS:#1}}
\newcommand{\refT}[1]{Theorem~\ref{T:#1}}
\newcommand{\refO}[1]{Observation~\ref{O:#1}}
\newcommand{\refP}[1]{Proposition~\ref{P:#1}}
\newcommand{\refD}[1]{Definition~\ref{D:#1}}
\newcommand{\refC}[1]{Corollary~\ref{C:#1}}
\newcommand{\refL}[1]{Lemma~\ref{L:#1}}
\newcommand{\refEx}[1]{Example~\ref{Ex:#1}}
\newcommand{\R}{\ensuremath{\mathbb{R}}}
\newcommand{\C}{\ensuremath{\mathbb{C}}}
\newcommand{\N}{\ensuremath{\mathbb{N}}}
\newcommand{\Q}{\ensuremath{\mathbb{Q}}}
\renewcommand{\P}{\ensuremath{\mathcal{P}}}
\newcommand{\Z}{\ensuremath{\mathbb{Z}}}
\newcommand{\kv}{{k^{\vee}}}
\renewcommand{\l}{\lambda}
\newcommand{\gb}{\overline{\mathfrak{g}}}
\newcommand{\dt}{\tilde d}     
\newcommand{\hb}{\overline{\mathfrak{h}}}
\newcommand{\g}{\mathfrak{g}}
\newcommand{\h}{\mathfrak{h}}
\newcommand{\gh}{\widehat{\mathfrak{g}}}
\newcommand{\ghN}{\widehat{\mathfrak{g}_{(N)}}}
\newcommand{\gbN}{\overline{\mathfrak{g}_{(N)}}}
\newcommand{\tr}{\mathrm{tr}}
\newcommand{\gln}{\mathfrak{gl}(n)}
\newcommand{\son}{\mathfrak{so}(n)}
\newcommand{\spnn}{\mathfrak{sp}(2n)}
\newcommand{\sln}{\mathfrak{sl}}
\newcommand{\sn}{\mathfrak{s}}
\newcommand{\so}{\mathfrak{so}}
\newcommand{\spn}{\mathfrak{sp}}
\newcommand{\tsp}{\mathfrak{tsp}(2n)}
\newcommand{\gl}{\mathfrak{gl}}
\newcommand{\slnb}{{\overline{\mathfrak{sl}}}}
\newcommand{\snb}{{\overline{\mathfrak{s}}}}
\newcommand{\sob}{{\overline{\mathfrak{so}}}}
\newcommand{\spnb}{{\overline{\mathfrak{sp}}}}
\newcommand{\glb}{{\overline{\mathfrak{gl}}}}
\newcommand{\Hwft}{\mathcal{H}_{F,\tau}}
\newcommand{\Hwftm}{\mathcal{H}_{F,\tau}^{(m)}}

\newcommand{\car}{{\mathfrak{h}}}    
\newcommand{\bor}{{\mathfrak{b}}}    
\newcommand{\nil}{{\mathfrak{n}}}    
\newcommand{\vp}{{\varphi}}
\newcommand{\bh}{\widehat{\mathfrak{b}}}  
\newcommand{\bb}{\overline{\mathfrak{b}}}  
\newcommand{\Vh}{\widehat{\mathcal V}}
\newcommand{\KZ}{Kniz\-hnik-Zamo\-lod\-chi\-kov}
\newcommand{\TUY}{Tsuchia, Ueno  and Yamada}
\newcommand{\KN} {Kri\-che\-ver-Novi\-kov}
\newcommand{\pN}{\ensuremath{(P_1,P_2,\ldots,P_N)}}
\newcommand{\xN}{\ensuremath{(\xi_1,\xi_2,\ldots,\xi_N)}}
\newcommand{\lN}{\ensuremath{(\lambda_1,\lambda_2,\ldots,\lambda_N)}}
\newcommand{\iN}{\ensuremath{1,\ldots, N}}
\newcommand{\iNf}{\ensuremath{1,\ldots, N,\infty}}

\newcommand{\tb}{\tilde \beta}
\newcommand{\tk}{\tilde \varkappa}
\newcommand{\ka}{\kappa}
\renewcommand{\k}{\varkappa}
\newcommand{\ce}{{c}}

\newcommand{\Pif} {P_{\infty}}
\newcommand{\Pinf} {P_{\infty}}
\newcommand{\PN}{\ensuremath{\{P_1,P_2,\ldots,P_N\}}}
\newcommand{\PNi}{\ensuremath{\{P_1,P_2,\ldots,P_N,P_\infty\}}}
\newcommand{\Fln}[1][n]{F_{#1}^\lambda}
\newcommand{\tang}{\mathrm{T}}
\newcommand{\Kl}[1][\lambda]{\can^{#1}}
\newcommand{\A}{\mathcal{A}}
\newcommand{\U}{\mathcal{U}}
\newcommand{\V}{\mathcal{V}}
\newcommand{\W}{\mathcal{W}}
\renewcommand{\O}{\mathcal{O}}
\newcommand{\Ae}{\widehat{\mathcal{A}}}
\newcommand{\Ah}{\widehat{\mathcal{A}}}
\newcommand{\La}{\mathcal{L}}
\newcommand{\Le}{\widehat{\mathcal{L}}}
\newcommand{\Lh}{\widehat{\mathcal{L}}}
\newcommand{\eh}{\widehat{e}}
\newcommand{\Da}{\mathcal{D}}
\newcommand{\kndual}[2]{\langle #1,#2\rangle}
\newcommand{\cins}{\frac 1{2\pi\mathrm{i}}\int_{C_S}}
\newcommand{\cinsl}{\frac 1{24\pi\mathrm{i}}\int_{C_S}}
\newcommand{\cinc}[1]{\frac 1{2\pi\mathrm{i}}\int_{#1}}
\newcommand{\cintl}[1]{\frac 1{24\pi\mathrm{i}}\int_{#1 }}
\newcommand{\w}{\omega}
\newcommand{\ord}{\operatorname{ord}}
\newcommand{\res}{\operatorname{res}}
\newcommand{\nord}[1]{:\mkern-5mu{#1}\mkern-5mu:}
\newcommand{\codim}{\operatorname{codim}}
\newcommand{\ad}{\operatorname{ad}}
\newcommand{\Ad}{\operatorname{Ad}}
\newcommand{\supp}{\operatorname{support}}

\newcommand{\Fn}[1][\lambda]{\mathcal{F}^{#1}}
\newcommand{\Fl}[1][\lambda]{\mathcal{F}^{#1}}
\renewcommand{\Re}{\mathrm{Re}}

\newcommand{\ha}{H^\alpha}

\define\ldot{\hskip 1pt.\hskip 1pt}
\define\ifft{\qquad\text{if and only if}\qquad}
\define\a{\alpha}
\redefine\d{\delta}
\define\w{\omega}
\define\ep{\epsilon}
\redefine\b{\beta} \redefine\t{\tau} \redefine\i{{\,\mathrm{i}}\,}
\define\ga{\gamma}
\define\cint #1{\frac 1{2\pi\i}\int_{C_{#1}}}
\define\cintta{\frac 1{2\pi\i}\int_{C_{\tau}}}
\define\cintt{\frac 1{2\pi\i}\oint_{C}}
\define\cinttp{\frac 1{2\pi\i}\int_{C_{\tau'}}}
\define\cinto{\frac 1{2\pi\i}\int_{C_{0}}}
\define\cinttt{\frac 1{24\pi\i}\int_C}
\define\cintd{\frac 1{(2\pi \i)^2}\iint\limits_{C_{\tau}\,C_{\tau'}}}
\define\dintd{\frac 1{(2\pi \i)^2}\iint\limits_{C\,C'}}
\define\cintdr{\frac 1{(2\pi \i)^3}\int_{C_{\tau}}\int_{C_{\tau'}}
\int_{C_{\tau''}}}
\define\im{\operatorname{Im}}
\define\re{\operatorname{Re}}
\define\res{\operatorname{res}}
\redefine\deg{\operatornamewithlimits{deg}}
\define\ord{\operatorname{ord}}
\define\rank{\operatorname{rank}}
\define\fpz{\frac {d }{dz}}
\define\dzl{\,{dz}^\l}
\define\pfz#1{\frac {d#1}{dz}}

\define\K{\Cal K}
\define\U{\Cal U}
\redefine\O{\Cal O}
\define\He{\text{\rm H}^1}
\redefine\H{{\mathrm{H}}}
\define\Ho{\text{\rm H}^0}
\define\A{\Cal A}
\define\Do{\Cal D^{1}}
\define\Dh{\widehat{\mathcal{D}}^{1}}
\redefine\L{\Cal L}
\newcommand{\ND}{\ensuremath{\mathcal{N}^D}}
\redefine\D{\Cal D^{1}}
\define\KN {Kri\-che\-ver-Novi\-kov}
\define\Pif {{P_{\infty}}}
\define\Uif {{U_{\infty}}}
\define\Uifs {{U_{\infty}^*}}
\define\KM {Kac-Moody}
\define\Fln{\Cal F^\lambda_n}
\define\gb{\overline{\mathfrak{ g}}}
\define\G{\overline{\mathfrak{ g}}}
\define\Gb{\overline{\mathfrak{ g}}}
\redefine\g{\mathfrak{ g}}
\define\Gh{\widehat{\mathfrak{ g}}}
\define\gh{\widehat{\mathfrak{ g}}}
\define\Ah{\widehat{\Cal A}}
\define\Lh{\widehat{\Cal L}}
\define\Ugh{\Cal U(\Gh)}
\define\Xh{\hat X}
\define\Tld{...}
\define\iN{i=1,\ldots,N}
\define\iNi{i=1,\ldots,N,\infty}
\define\pN{p=1,\ldots,N}
\define\pNi{p=1,\ldots,N,\infty}
\define\de{\delta}

\define\kndual#1#2{\langle #1,#2\rangle}
\define \nord #1{:\mkern-5mu{#1}\mkern-5mu:}
\newcommand{\MgN}{\mathcal{M}_{g,N}} 
\newcommand{\MgNeki}{\mathcal{M}_{g,N+1}^{(k,\infty)}} 
\newcommand{\MgNeei}{\mathcal{M}_{g,N+1}^{(1,\infty)}} 
\newcommand{\MgNekp}{\mathcal{M}_{g,N+1}^{(k,p)}} 
\newcommand{\MgNkp}{\mathcal{M}_{g,N}^{(k,p)}} 
\newcommand{\MgNk}{\mathcal{M}_{g,N}^{(k)}} 
\newcommand{\MgNekpp}{\mathcal{M}_{g,N+1}^{(k,p')}} 
\newcommand{\MgNekkpp}{\mathcal{M}_{g,N+1}^{(k',p')}} 
\newcommand{\MgNezp}{\mathcal{M}_{g,N+1}^{(0,p)}} 
\newcommand{\MgNeep}{\mathcal{M}_{g,N+1}^{(1,p)}} 
\newcommand{\MgNeee}{\mathcal{M}_{g,N+1}^{(1,1)}} 
\newcommand{\MgNeez}{\mathcal{M}_{g,N+1}^{(1,0)}} 
\newcommand{\MgNezz}{\mathcal{M}_{g,N+1}^{(0,0)}} 
\newcommand{\MgNi}{\mathcal{M}_{g,N}^{\infty}} 
\newcommand{\MgNe}{\mathcal{M}_{g,N+1}} 
\newcommand{\MgNep}{\mathcal{M}_{g,N+1}^{(1)}} 
\newcommand{\MgNp}{\mathcal{M}_{g,N}^{(1)}} 
\newcommand{\Mgep}{\mathcal{M}_{g,1}^{(p)}} 
\newcommand{\MegN}{\mathcal{M}_{g,N+1}^{(1)}} 

\define \sinf{{\widehat{\sigma}}_\infty}
\define\Wt{\widetilde{W}}
\define\St{\widetilde{S}}
\newcommand{\SigmaT}{\widetilde{\Sigma}}
\newcommand{\hT}{\widetilde{\frak h}}
\define\Wn{W^{(1)}}
\define\Wtn{\widetilde{W}^{(1)}}
\define\btn{\tilde b^{(1)}}
\define\bt{\tilde b}
\define\bn{b^{(1)}}
\define \ainf{{\frak a}_\infty} 

%
\define\eps{\varepsilon}    
\newcommand{\e}{\varepsilon}
\define\doint{({\frac 1{2\pi\i}})^2\oint\limits _{C_0}
       \oint\limits _{C_0}}                            
\define\noint{ {\frac 1{2\pi\i}} \oint}   
\define \fh{{\frak h}}     
\define \fg{{\frak g}}     
\define \GKN{{\Cal G}}   
\define \gaff{{\hat\frak g}}   
\define\V{\Cal V}
\define \ms{{\Cal M}_{g,N}} 
\define \mse{{\Cal M}_{g,N+1}} 
\define \tOmega{\Tilde\Omega}
\define \tw{\Tilde\omega}
\define \hw{\hat\omega}
\define \s{\sigma}
\define \car{{\frak h}}    
\define \bor{{\frak b}}    
\define \nil{{\frak n}}    
\define \vp{{\varphi}}
\define\bh{\widehat{\frak b}}  
\define\bb{\overline{\frak b}}  
\define\KZ{Knizhnik-Zamolodchikov}
\define\ai{{\alpha(i)}}
\define\ak{{\alpha(k)}}
\define\aj{{\alpha(j)}}
\newcommand{\calF}{{\mathcal F}}
\newcommand{\ferm}{{\mathcal F}^{\infty /2}}
\newcommand{\Aut}{\operatorname{Aut}}
\newcommand{\End}{\operatorname{End}}
\newcommand{\laxgl}{\overline{\mathfrak{gl}}}
\newcommand{\laxsl}{\overline{\mathfrak{sl}}}
\newcommand{\laxso}{\overline{\mathfrak{so}}}
\newcommand{\laxsp}{\overline{\mathfrak{sp}}}
\newcommand{\laxs}{\overline{\mathfrak{s}}}
\newcommand{\laxg}{\overline{\frak g}}
\newcommand{\bgl}{\laxgl(n)}
\newcommand{\tX}{\widetilde{X}}
\newcommand{\tY}{\widetilde{Y}}
\newcommand{\tZ}{\widetilde{Z}}

\title[]{Inverse scattering method for Hitchin systems of types $B_n$, $C_n$, $D_n$, and their generalizations}
\author[O.K.Sheinman]{O.K.Sheinman}
\dedicatory{}
\maketitle
\begin{abstract}
We  give a solution of the Inverse Scattering Problem for integrable systems with a finite number degrees of freedom, admitting a Lax representation with spectral parameter on a Riemann surface. While conventional approaches deal with the systems with $GL(n)$ symmetry, we focus on the problems arising in the case of symmetry with respect to a semi-simple group. Our main results apply to Hitchin systems of the types $B_n$, $C_n$, $D_n$.
\end{abstract}
\tableofcontents
\section{Introduction}
In \cite{Klax}, I.M.Krichever introduced a wide class of Lax integrable systems with spectral parameter on a Riemann surface. Basic examples are given by $A_n$-type Hitchin and Calogero--Moser systems. Some classical integrable systems are contained there as well. In the above quoted work, Krichever has constructed the corresponding integrable hierarchies and their Hamiltonian theory, gave a general scheme of algebraic-geometrical integration of $GL(n)$ Hitchin systems by means his well-known theory of Baker--Akhieser functions.

In \cite{Kr_Sh_FAN}, there were constructed similar Lax operators for the classical groups $SO(2n)$, $SO(2n+1)$, $Sp(2n)$.  In the subsequent series of works summarized in \cite{Sh_DGr} we generalized the hierarchies originally introduced in \cite{Klax}, and their Hamiltonian theory, on the Lax systems of the above type where operators of the $(L,M)$-pair take values in the corresponding simple Lie algebras. In \cite{Shein_UMN2016} we generalized the above results on the case of arbitrary semi-simple Lie algebras. It turned out to be possible due to a new algebraic approach to the Lax operators in question developed in \cite{Sh_TG2016}. In \cite{Sh_FAN_2019,Sh_Bor}, we developed explicit separation of variables technique for Hitchin systems on hyperelliptic curves with arbitrary underlying symmetry algebra (group) in the class of complex semi-simple Lie algebras (groups, resp.).

The inverse scattering method appears as the only part of the programme initiated in \cite{Klax}, not yet done in presence of non-trivial group symmetry. It is the aim of this work to fill this gap.

In the present paper we consider the case of classical simple groups $SO(2n)$, $SO(2n+1)$, $Sp(2n)$. We set the problem as follows: to retrieve the operators of the Lax pair by spectral data. With this aim, we use a conventional scheme of algebraic-geometric integration due to Krichever \cite{Kr_UMN1977,Krichever_Itogi}, specified for the systems in question in \cite{Klax}. According to this scheme
\begin{equation}\label{E:LM}
      L=\widehat\Psi\Lambda{\widehat\Psi}^{-1},\ M=-\partial_t\widehat\Psi\cdot{\widehat\Psi}^{-1}
\end{equation}
where $L$ is the Lax matrix, $M$ is its counterpart in the Lax pair, $\Lambda$ is a spectrum of $L$ coming directly from the spectral curve, ${\widehat\Psi}$ is the matrix of Baker--Akhieser functions uniquely defined by spectral data.

In presence of a symmetry with respect to a semi-simple group $G$, there is the following difference with the conventional $GL(n)$ case :
\begin{itemize}
\item[1)]
${\widehat\Psi}$ is a $G$-valued function;

\item[2)]
Conventional spectral data is not sufficient for retrieving the $(L,M)$-pair. It is necessary to give explicitly the time-dependent poles of ${\widehat\Psi}$ corresponding to the \emph{Tyurin points} of the Lax matrix while in the conventional $GL(n)$ case the Tyurin points are zeroes of the Baker--Akhieser function.

\item[3)]
An expression for exponents of the Baker--Akhieser function is needed.
\end{itemize}
None of the problems 1, 2 arose in the $GL(n)$ case, and the problem 3 had a trivial solution. We resolve the first problem by applying a certain orthogonalization (resp., skew-orthogonalization) process to the matrix ${\widehat\Psi}$ obtained in a conventional way. The corresponding bilinear form comes directly from the assumption of invariance of the spectral curve with respect to a holomorphic involution (in turn, coming from the assumption that $L$ is a skew-symmetric (resp., infinitesimal symplectic)) matrix. The most delicate part of the proof is to preserve the form of Tyurin poles in course of the orthogonalization (proof of the \refT{exfor} below).

Methods of resolving the second problem are actually developed in \cite{KN_UMN1980}. We adopt the equation of motion of our analogs of the Tyurin parameters from \cite{Shein_UMN2016}.

To resolve the third problem, we reduce it to a calculation of gradients of invariants of the $L$-matrix in terms of its spectrum. This looks like a classical problem of linear algebra, however we don't know any reference on its solution.

In \refS{Lax_int_sys} we define integrable hierarchies in question via their Lax representations.

In \refS{sp_tr} we define the spectral data corresponding to a given Lax hierarchy. They include the spectral curve, the pole divisor and divisor of exponential singularities of the Baker--Akhieser function, and the form of the corresponding exponents. We stress again that these data are insufficient in the case $G$ is semi-simple. The reason is that in the semi-simple case the Tyurin points appear as unremovable poles while for $G=GL(n)$ they are zeroes of (determinant of) the Baker--Akhieser function.

In \refS{i_sp_tr} we formulate and prove main results of the paper, namely, the theorem of existence and uniqueness of the Baker--Akhieser function possessing the above listed special properties, and a theta-functional formula for it. We also give an explicit formula for the exponents of the Baker--Akhieser function.

\section{Lax integrable systems with spectral parameter on a Riemann surface}\label{S:Lax_int_sys}
\subsection{Lax equations}\label{SS:Lax_eq}
Let $\g$ be a semi-simple Lie algebra over $\C$, $\h$ be its Cartan subalgebra, and $h\in\h$ be such element that  $p_i=\a_i(h)\in\Z_+$ for every simple root $\a_i$ of $\g$. If we denote the root lattice of $\g$ by $\Z(R)$ where $R$ is the root system of $\g$, then $h$ belongs to the positive chamber of the dual lattice $\Z(R)^*$.

For $p\in\Z$ let $\g_p=\{ X\in\g\ |\ (\ad h)X=pX \}$, and $k=\max\{p\ |\ \g_p\ne 0\}$. Then the decomposition $\g=\bigoplus\limits_{i=-k}^{k}\g_p$ gives a $\Z$-grading on $\g$. Call $k$ a \emph{depth} of the grading. Obviously,
\[
 \g_p=\bigoplus\limits_{\substack{\a\in R\\ \a(h)=p}}\g_\a .
\]
Define also the following filtration on $\g$:  $\tilde\g_p=\bigoplus\limits_{q=-k}^p\g_q$. Then $\tilde\g_p\subset\tilde\g_{p+1}$ ($p\ge -k$), $\tilde\g_{-k}=\g_{-k},\ldots,\tilde\g_k=\g$, $\tilde\g_p=\g$, $p>k$.

Let $\Sigma$ be a complex compact Riemann surface with two given non-intersecting finite sets of marked points: $\Pi$ and $\Gamma$. Assume every $\ga\in\Gamma$ to be assigned with an $h_\ga\in \Z(R)^*_+$, and with the corresponding grading and filtration. We equip the notation $\g_p$, $\tilde\g_p$ with the upper $\ga$ indicating that the grading (resp. filtration) subspace corresponds to $\ga$.  Let $L$ be a meromorphic mapping $\Sigma\to\g$, holomorphic outside the marked points which may have poles of an arbitrary order at the points in $\Pi$, and has the expansion of the following form in a neighborhood of any~$\ga\in\Gamma$:
\begin{equation}\label{E:ga_expan}
   L(z)=\sum\limits_{p=-k}^\infty L_p(z-z_\ga)^p,\ L_p\in\tilde\g_p^\ga
\end{equation}
where $z$ is a local coordinate in the neighborhood of $\ga$, $z_\ga$ is the coordinate of $\ga$ itself. For simplicity, we assume that the depth of grading $k$ is the same all over $\Gamma$, though it would be no difference otherwise.
\begin{lemma}[\cite{Shein_UMN2016}]\label{L:lform}
The expansion \refE{ga_expan} takes place if, and only if \[
  L(z)=(z-z_\ga)^{h_\ga}L_{(0)}(z)(z-z_\ga)^{-h_\ga}
\]
where  $L_{(0)}$ is holomorphic in the neighborhood of $z=z_\ga$.
\end{lemma}
We denote by $\L_{\Pi,\Gamma,h}$ a linear space of all such mappings. Since the relation \refE{ga_expan} is preserved under commutator, $\L_{\Pi,\Gamma,h}$ is a Lie algebra called {\it Lax operator algebra}.

Let $M:\ \Sigma\to\g$ be a meromorphic mapping  holomorphic outside $\Pi$ and $\Gamma$, for every $\ga\in\Gamma$ having a Laurent expansion
\begin{equation}\label{E:M_oper}
  M(z)=\frac{\nu_\ga h_\ga}{z-z_\ga}+\sum_{i=-k}^\infty M_i^\ga(z-z_\ga)^i
\end{equation}
at $\ga$, where $M_i^\ga\in\tilde\g_i^\ga$ for $i<0$, $M_i^\ga\in\g$ for $i\ge 0$, and $\nu_\ga\in\C$. We denote by $\M_{\Pi,\Gamma,h}$ the space of such mapping corresponding to given data $\Pi,\Gamma,h$ (as above, $h=\{ h_\ga\ |\ \ga\in\Gamma\}$).  Obviously, $\L_{\Pi,\Gamma,h}\subset\M_{\Pi,\Gamma,h}$.

Let $D$ be a non-negative divisor supported at $\Pi$: $D=m_1P_1+\ldots+m_NP_N$ where $m_1,\ldots,m_N\ge 0$. Let $\L_{\Pi,\Gamma,h}^D=\{ L\in\L_{\Pi,\Gamma,h}\ |\ (L)+\sum\limits_{s=1}^K k_s\ga_s+D\ge 0 \}$. We define a sheaf $\L^D$ formed by the spaces $\L_{\Pi,\Gamma,h}^D$ as fibres.  A point of the total space of the sheaf is given by a triple $\{\Gamma,h,L\}$ where the pair $\{ \Gamma,h \}$ ($h=\{ h_\ga\ |\ \ga\in\Gamma  \}$) represents a point of the base, $L\in\L_{\Pi,\Gamma,h}$ represents a point of the fibre over $\{\Gamma,h\}$.

We assume $\g$ to be a classical simple Lie algebra, and $G$ be the corresponding classical group. Then $G$ operates on the total space of the sheaf $\L^D$. The action of a $g\in G$ is as follows: $\ga\to\ga$, $h_\ga\to gh_\ga g^{-1}$, $L\to gL g^{-1}$.

Many results below apply to a more general situation of an arbitrary semi-simple Lie algebra $\g$ and the algebraic group $G$ corresponding to $\g$ and its faithful representation in a linear space $V$.
\begin{definition}\label{D:PhSp}
The quotient of the total space of the sheaf $\L^D$ by the just defined $G$-action serves as a \emph{phase space} of the dynamical system we are going to define. Let this space be denoted by $\P^D$. Thus $\P^D=\L^D/G$.
\end{definition}

We give the dynamics on the phase space by means of a Lax equation. It is our next step to introduce it.

By \emph{Lax operator} we mean the map of the total space of the sheaf $\L^D$ to $Mer(\Sigma\to\g)$ given by $\{\Gamma,h,L\}\to L$, i.e. the map forgetting the base component of a point in $\L^D$. Taking account of $G$-equivariance of the map, we regard to it as to a map defined on the phase space. In abuse of notation, we will denote the map (i.e. the Lax operator) by~$L$ also.

We introduce the sheaf $\M^D$ similarly to $\L^D$ replacing $L$ defined by \refE{ga_expan} with $M$ defined by \refE{M_oper} in the above definitions. We call the forgetting map on $\M^D :\, \{ \Gamma, h, M\}\to M$ the \emph{$M$-operator}.

The equation
\begin{equation}\label{E:Lax_eq}
 {\dot L} = [L,M]
\end{equation}
where ${\dot L}= dL/dt$, is called the \emph{Lax equation}. This is a system of ordinary differential equations for the dynamical variables including $\{ (z_\ga, h_\ga)\ |\ \ga\in\Gamma\}$, and the parameters of the main parts of meromorphic functions $L$ and $M$ at the points in $\Pi$ and $\Gamma$.

Following \cite{Klax}, a system given by the Lax representation \refE{ga_expan}--\refE{Lax_eq} with $D=(\w)$ where $\w$ is a holomorphic differential on $\Sigma$, is called a \emph{Hitchin system}.

\subsection{Hierarchies}
In order that the system \refE{Lax_eq} is closed, it is necessary to give $M$ as a function of $L$. Here we will do it following the lines of \cite{Klax,Shein_UMN2016,Sh_DGr}.

Let $a$ denote a triple of the form $\{ \chi, P\in\Pi, m>-m_P\}$ where $\chi$ is an invariant polynomial on the Lie algebra $\g$, $m_P$ is the mulniplicity of $P$ in the divisor $D$.

We define the \emph{gradient} $\nabla \chi(L)\in\g$ of the polynomial $\chi$ at the point $L\in\g$ by means of the relation
\begin{equation}\label{E:variation}
   d\chi(L)=\langle \nabla \chi(L),d L\rangle ,
\end{equation}
where $d\chi$ is the differential of $\chi$ as of a function on $\g$, $\langle \cdot\, ,\cdot\rangle$ is the Cartan--Killing form on $\g$. If $L\in\L_{\Pi,\Gamma,h}$, i.e. we regard to $L$ as to a meromorphic $\g$-valued function on $\Sigma$, then $\nabla \chi(L)$ will be such function too. If it is considered as a function of a local coordinate $w$ on $\Sigma$, we write $\nabla \chi(w)$.
\begin{lemma}[\cite{Klax,Shein_UMN2016}]\label{L:Ma}
For every triple $a=\{ \chi, P, m\}$, and every $L\in\L_{\Pi,\Gamma,h}^D$, there is a unique appropriately normalized $M_a\in\M_{\Pi,\Gamma,h}^D$ having a unique pole outside $\Gamma$, namely, at $P$, and such that in a neighborhood of $P$ (with a local parameter $\w$)
\begin{equation}\label{E:sravn}
     M_a(w)=w^{-m}\nabla \chi(L(w))+O(1),
\end{equation}
and the equation ${\dot L} = [L,M_a]$ is well-defined (i.e. $([L,M_a])$ is a tangent vector to $\L_{\Pi,\Gamma,h}^D$.
\end{lemma}
\refL{Ma} gives $M_a$ as a function of $L$ which we denote by $M_a(L)$. We refer to \cite{Shein_UMN2016} for the details on the normalization mentioned in the theorem.

A point $P\in\Sigma$ is called a \emph{regular} point of $L$ if $L(P)$ is well-defined, and is a regular element of the Lie algebra $\g$. It is called  \emph{irregular} if this evaluation is well-defined but irregular in $\g$. The set of all irregular points of a given $L\in\L_{\Pi,\Gamma,h}^D$ is finite. This implies that the set of Lax operators for which the set of irregular points has an empty intersection with $\Pi$, is open.
\begin{theorem}[\cite{Klax,Shein_UMN2016}]\label{T:hierarch}
Relations
\begin{equation}\label{E:part_a}
      \partial_aL=[L,M_a]
\end{equation}
where $M_a=M_a(L)$, give a family of commuting vector fields (parameterized by $a$) on the open subset in $\P^D$ given by Lax operators such that their sets of irregular points have empty intersections with $\Pi$.
\end{theorem}
\begin{corollary}\label{C:coroll}
Lax equations $\partial_a L=[L,M_a]$ give commuting flows on $\P^D$.
\end{corollary}
In \cite{Shein_UMN2016}, \refT{hierarch} has been proven under conditions \cite[(3.21)]{Shein_UMN2016} which read as follows in our notation
\begin{equation}\label{E:mov_tyur}
   \partial_a z_\ga=\nu_{a,\ga},\ \partial_ah_\ga=[h_\ga,M_{a,0}^\ga]
\end{equation}
($\nu_{a,\ga}$ coming from \refE{M_oper}).
The same relations appear as the conditions of holomorphy of spectra of $M$-operators along trajectories of \refE{Lax_eq} \cite[Section 4.5]{Shein_UMN2016}. In our approach, equations \refE{mov_tyur} play the same role as the \emph{equations of motion of Tyurin parameters} in \cite{KN_UMN1980,Klax}.
\begin{theorem}
For any $a$ the vector field $\partial_a$ is Hamiltonian with the Hamiltonian
\[
        H_a=\res_{P_a}w_a^{-m_a}\chi(L(w_a))\w(w_a)
\]
with respect to the Krichever--Phong symplectic structure.
\end{theorem}
We refer to \cite{Shein_UMN2016} for the version of the Krichever--Phong symplectic structure needed here. For the origin of this notion we refer to \cite{Klax} and references therein, and to \cite{Sh_DGr}.
\section{Spectral transform}\label{S:sp_tr}
It is the aim of the section, to assign a certain spectral data to the above constructed hierarchy of Lax equations.
\subsection{Spectral curve}\label{SS:spec_curve}
By the {\it spectral curve} of an $L\in\P^D$ we mean a curve $\Sigma_L$ given by the equation
\begin{equation}\label{E:spec_curve}
   \det(L(q)-\l)=0,\ q\in\Sigma.
\end{equation}
It is a $d$-fold branch covering of $\Sigma$ where $d=\dim V$, $V$ is the space of the standard representation of $G$.

The spectral curve possesses the following simple properties.
\begin{lemma}\label{L:prop_spec}
\begin{itemize}
\item[$1^\circ$]
Represent \refE{spec_curve} in the form
\begin{equation}\label{E:sp_c_eq}
    R(q,\l)=\l^{d}+\sum_{j=1}^d r_j(q)\l^{d-j}=0.
\end{equation}
Then $r_j\in {\mathcal O}(\Sigma,-jD)$.

\item[$2^\circ$]
Assume $\g$ is one of the Lie algebras $\so(2n)$, $\so(2n+1)$, $\spn(2n)$. Then the spectral curve is invariant with respect to the involution $\l\to -\l$.
\end{itemize}
\end{lemma}
\begin{proof}
$1^\circ$ According to \cite{Klax,Shein_UMN2016,Sh_DGr} the spectra of Lax  operators are holomorphic at Tyurin points. Hence polynomials $r_j(q)$ are holomorphic there too. Next, $(L)\ge -D$. Since $r_j$ is a polynomial of degree $j$ in $L$, we obtain $(r_j)\ge -jD$.

$2^\circ$
Let $\s$ be a matrix of the scalar product in $V$ in the cases $\g=\so(2n),\so(2n+1)$, and a matrix of the symplectic form in the case $\g=\spn(2n)$. Then $\s L^T\s^{-1}=-L$ ($L^T$ denotes the matrix transposed to $L$), and \refE{spec_curve} is obviously invariant with respect to the replacement of $L$ with $\s L^T\s^{-1}$. But this replacement results in the equation $\det(L(q)+\l)=0$.
\end{proof}
Let $D_L$ be a pull-back of the divisor $D$ with respect of the covering $\Sigma_L\to\Sigma$.

\subsection{Relation between spectra of $L$- and $M$-operators}
Since $L$ and $\partial_a+M_a$ commute by virtue of the Lax equations, there is a meromorphic vector-valued function $\psi$ on $\Sigma_L$ such that for a general $Q=(q,\l)\in\Sigma_L$
\begin{equation}\label{E:eigen_up}
        L(q)\psi(Q)=\l(Q)\psi(Q),\ (\partial_a+M_a(q))\psi(Q)=f_a(Q)\psi(Q)
\end{equation}
where $\l(Q)$ is a solution of \refE{spec_curve}. For a given $q\in\Sigma$, let $\Lambda(q)$ be a diagonal matrix with all solutions of \refE{spec_curve} on the diagonal. To define it, we fix an arbitrary order of the sheets of the covering $\Sigma_L\to\Sigma$, that is an order of the roots of the equation \refE{spec_curve}. As well, we consider a diagonal matrix $F_a(q)$ with the eigenvalues $f_a(q)$ on the diagonal, placed in the same order.

Let $\Lambda(q)=diag(\l_1(q),\ldots,\l_d(q))$. Placing the vectors $\psi(q,\l_1),\ldots,\psi(q,\l_d)$ in the corresponding order, we obtain a matrix $\Psi(q)$ which we will call \emph{the matrix of eigenvectors}. The following relations are equivalent to \refE{eigen_up}
\begin{equation}\label{E:eigen_down}
     L\Psi=\Psi\Lambda,\ (\partial_a+M_a)\Psi=\Psi F_a.
\end{equation}

The two matrices $\Lambda$ and $F_a$ represent the spectra of the operators $L$ and $M_a$, resp. The relation between the spectra in a neighborhood of $q=P_a$ is given by
\begin{equation}\label{E:relsp1}
      F_a(w_a)=w_a^{-m_a}\nabla\chi_a(\Lambda(w_a))+O(1)
\end{equation}
where $w_a$ is a local parameter in the neighborhood. It has been proven in \cite{Shein_UMN2016}. Indeed, in a neighborhood of $P_a$ we have, first, $F_a=\Psi^{-1}\cdot\partial_a\Psi+\Psi^{-1} M_a\Psi=\Psi^{-1} M_a\Psi+O(1)$ by holomorphy of $\Psi$ and $\Psi^{-1}$ at the points in $\Pi$ (which is a condition of a general position), and second, $M_a = w^{-m}\nabla\chi(L(w))+O(1)$ by \refL{Ma}. By invariance of $\chi$, its gradient $\nabla\chi$ is equivariant, thus we have
$\Psi\nabla\chi(L)\Psi^{-1}=\nabla\chi(\Psi L\Psi^{-1})=\nabla\chi(\Lambda)$ (the first equality follows from \cite[Eq. (3.29)]{Shein_UMN2016}). Therefore, $F_a = \nabla\chi(\Lambda) + O(1)$.

\subsection{From $(L,M)$-pair to Baker--Akhiezer matrix-function}
Given a pair $(L,M_a)$ defined as above, consider the following matrix-valued function on $\Sigma$:
\begin{equation}\label{E:psiexp}
      {\widehat\Psi}=\Psi\exp\left( -\int\limits^{t_a} F_a dt_a   \right).
\end{equation}
where $t_a$ is the time parameter of the flow given by $M_a$.

\begin{theorem}\label{T:prep_0}
\begin{itemize}
\item[(A0)]
$(\partial_a+M_a){\widehat\Psi}=0$.
\item[(A1)]
Columns of $\widehat\Psi$ are orthogonal if $\g=\so(2n),\so(2n+1)$, and skew-orthogonal if $\g=\spn(2n)$.
\item[(A2)]
${\widehat\Psi}$ is meromorphic except at $P_a$ where it has an essential singularity of the form
\[
\widehat\Psi(w_a)=\Psi_a(w_a){\exp}\left(- t_a w^{-m_a}_a\nabla\chi_a(\Lambda(w_a))\right),
\]
$w_a$ is a local parameter in the neighborhood, $\Psi_a$ is a local matrix-valued function, holomorphic and invertible, $\Lambda(w_a)$ is a local notation for $\Lambda(q)$.
\item[(A3)]
$\widehat\Psi$ is meromorphic outside $D$. Apart from Tyurin points, the divisor $D_p$ of its poles is time independent.
\item[(A4)]
At the Tyurin points
\begin{equation}\label{E:TPform_0}
    \widehat\Psi=(z-z_\ga)^{h_\ga}\Psi_0
\end{equation}
where $\Psi_0$ is holomorphic and holomorphically invertible in a neighborhood of $z_\ga$.
\end{itemize}
\end{theorem}
\begin{proof}
\item[(A0)]
We have
\begin{align*}
   (\partial_a+M_a){\widehat\Psi}&=\partial_a\Psi\cdot e^{-\int\limits^{t_a}F_adt_a} - \Psi F_ae^{-\int\limits^{t_a}F_adt_a} + M_a\Psi e^{-\int\limits^{t_a}F_adt_a} \\
   &= (\partial_a+M_a)\Psi e^{-\int\limits^{t_a}F_adt_a} - \Psi F_ae^{-\int\limits^{t_a}F_adt_a}.
\end{align*}
By \refE{eigen_down}, $(\partial_a+M_a)\Psi=\Psi F_a$, which completes the proof of the statement.

(A1) holds for the reason that $L$ is skew-symmetric in the case $\g=\so(2n),\so(2n+1)$, and infinitesimal symplectic in the case $\g=\spn(2n)$.

(A2-A3) We refer to \cite{Krichever_Itogi} for the proof of those properties which are conventional. The only point which is not standard is the form of exponent which immediately follows from \refE{relsp1}.

(A4)
By \refL{lform}, there takes place the following relation for the Lax operator:
\[
  L(z)=(z-z_\ga)^{h_\ga}L_{(0)}(z)(z-z_\ga)^{-h_\ga}
\]
where $L_{(0)}(z)$ is holomorphic. Let $\Psi_0(z)$ be the matrix formed by eigenvectors of $L_{(0)}(z)$. In a general position, it is holomorphic, and holomorphically invertible in a neighborhood of $z=z_\ga$. We have $L_{(0)}(z)\Psi_0(z)=\Psi_0(z)\Lambda$. Together with the previous relation it implies $L(z)(z-z_\ga)^{h_\ga}\Psi_0(z)=(z-z_\ga)^{h_\ga}\Psi_0(z)\Lambda$, hence $\Psi(z)=(z-z_\ga)^{h_\ga}\Psi_0(z)$ (up to normalization of eigenvectors which does not affect the proof).
\end{proof}
By the property (A2) $\widehat\Psi$ is a conventional \emph{Baker--Akhieser matrix}, as introduced in \cite{Krichever_Itogi,Klax}. However, by the property (A4), the Tyurin points appear as time-dependent poles. Indeed, $\tr\, h_\ga=0$, hence $(z-z_\ga)^{h_\ga}$ has both zeroes and poles on the diagonal at $z=z_\ga$ (while for $GL(n)$ $(z-z_\ga)^{h_\ga}$ is conjugated to $diag(z-z_\ga,1,\ldots, 1)$ \cite{Shein_UMN2016}). This makes a substantial difference with \cite{Klax}, though in general such situation is not new. It had been considered in \cite{KN_UMN1980} for the Kadomtsev--Petviashvili equation. Another difference with the conventional situation makes the property~(A1).

\subsection{Spectral transform}
By \emph{spectral transform} we mean the correspondence
\[
    \{ L,\{ M_a\}\}\longrightarrow \{ \Sigma_L, D, \{ (\ga, h_\ga) \},\{\nabla\chi_a(\Lambda)  \}  \}.
\]
In other words, the above hierarchy is assigned with the spectral data consisting of the spectral curve, the divisors $D$ (of exponential points of the Baker--Akhieser function) and $D_p$ (its pole divisor) on it, gradients of invariant polynomials, and the set of Tyurin points together with associated grading elements. We would like comment on the the last two items (Tyurin data and $\{\nabla\chi_a(\Lambda)  \}$). Integration of the Tyurin data in spectral data is not conventional. In the  precedential work \cite{Klax}, due to a special form of the elements $h_\ga$ for $\g=\gl(n)$, the Tyurin points turn out to be zeroes of eigenfunctions. The situation for semisimple algebras considered here is completely different: there are always poles at Tyurin points. As for the gradients of invariant polynomials, they are included as polynomials giving a form of the exponents which is conventional \cite{Krichever_Itogi}.

It is the aim of the next section to retrieve the hierarchy by the spectral data, in other words to construct the inverse spectral transform.

Inspite the above introduced $\widehat\Psi$ is not included to spectral data,   it gives a hint at which properties the Baker--Akhieser function giving the inverse spectral transform should possess.


\section{Inverse spectral transform}\label{S:i_sp_tr}
It is the aim of the section to retrieve the $L$- and $M$-operators by the spectral data and Tyurin data. We formulate and prove the theorem of existence and uniqueness of the Baker--Akhieser matrix possessing the above formulated properties (A1)--(A4) (\refT{prep_0}).

\subsection{Retrieving the $(L,M)$-pair by spectral and Tyurin data}\label{SS:KTH}
Given a genus $g$ curve $\Sigma$, and a positive divisor $D$ on it, let $\widehat\Sigma$ be a Riemann surface of the algebraic function $\l$ given by the equation
\begin{equation}\label{E:spkri}
    R(q,\l)=\l^{2n}+\sum_{j=1}^n r_j(q)\l^{2n-2j}=0
\end{equation}
where $r_j\in {\mathcal O}(\Sigma,-jD)$ (the left equality in \refE{spkri} is just a notation). $\widehat\Sigma$ is a branch covering of $\Sigma$. It possesses a holomorphic involution $\s :\, \l\to -\l$. Assume that $\widehat\Sigma$ has no more than nodal singularities. In the case $\widehat\Sigma$ is singular, let ${\widehat\Sigma}^n$ denote its normalization. Let   $\pi :\,\widehat\Sigma\to\Sigma$ ($\pi :\,{\widehat\Sigma}^n\to\Sigma$, resp.) be the covering map.

Let $\l(Q)$ be a function on $\widehat\Sigma$ defined by $\l(q,\l_j)=\l_j$ where $Q=(q,\l_j)$ is an arbitrary point of $\widehat\Sigma$. Let $\widehat D$ be a pull-back of the divisor $D$. Then $\widehat D$ is a positive divisor on $\widehat\Sigma$, obviously invariant with respect to $\s$. Let $\Lambda(q)$ be a diagonal matrix--valued function on $\Sigma$: $\Lambda(q)_{ij}=\d_{ij}\l_j$ where $\l_j$ is the $j$-th root of the spectral curve equation over the point $q\in\Sigma$ for an arbitrary (but fixed) ordering of the roots. By definition
\begin{equation}\label{E:lower_b}
   (\Lambda)+D\ge 0.
\end{equation}
The diagonal matrix-valued function $\Lambda$ is defined up to a preserving the involution enumeration of the roots of the characteristic equation, i.e. up to the action of the Weyl group on $\h$. Similarly we can assign a diagonal matrix-valued function on $\Sigma$ to any (scalar-valued) function on $\widehat\Sigma$. Once the enumeration is fixed, there is a version of the direct image construction which takes a vector-valued function on $\widehat\Sigma$ to a  matrix-valued function on $\Sigma$. The evaluation of the last at $q\in\Sigma$ is formed by the evaluations of the first at the preimages of $q$, as by columns, taken in the corresponding order.

Let $D=\sum_{P\in\Sigma} m_PP$ where $m_P\in \Z$, $m_P>0$. Let the time $t_a$ correspond to a point $P_a\in{\rm supp}\, D$, integer $m_a\ge -m_P$, and an invariant polynomial~$\chi_a$.
\begin{theorem}\label{T:BA_exist}
Given a non-special divisor $D_p$ on $\widehat\Sigma$, such that $\deg D_p={\widehat g}+d-1$ ($d$ is the dimension of the standard representation of $\g$), there exists the unique Baker--Akhieser vector-function $\widehat\psi$, and the corresponding matrix-valued function $\widehat\Psi$ satisfying the following conditions:
\begin{itemize}
\item[(BA1)]
$\widehat\Psi$ is a $G$-valued function on $\Sigma$.
\item[(BA2)]
For every $P\in{\rm supp}\, D$,
\[
\widehat\Psi(w_P)=\Psi_{P}(w_P)\cdot {\exp}\left(-\sum_{\substack{a:\, P_a=P,\\ m_a\le m_P} } t_a w^{-m_a}_P\nabla\chi_a(\Lambda(w_P))\right)
\]
in a punctured neighborhood of the point $P$, where $w_P$ is a local parameter in the neighborhood, $\Psi_P$ is a local matrix-valued function, holomorphic and invertible, $\Lambda(w_P)$ is a local notation for $\Lambda(q)$.
\item[(BA3)]
Outside $D$ ($\widehat D$, resp.), $\widehat\Psi$ ($\widehat\psi$, resp.) is meromorphic. Apart from Tyurin points, $(\widehat\psi)+D_p\ge 0$ (hence, the poles of $\widehat\psi$ are time independent).
\item[(BA4)]
At the Tyurin points
\begin{equation}\label{E:TPform}
    \widehat\Psi=(z-z_\ga)^{h_\ga}\Psi_0
\end{equation}
where $\Psi_0$ is holomorphic and holomorphically invertible in a neighborhood of $z_\ga$ (including $z_\ga$).
\end{itemize}
\end{theorem}
\begin{proof}
Let ${\mathbf p}=\{ (Q,p_Q)\ |\ Q\in\widehat\Sigma,\ p_Q\in\C[w_Q^{-1}]  \}$ be finite ($w_Q$ denotes a local parameter in a neighborhood of $Q$), $D_p$ be a positive divisor of degree $\widehat g+d-1$ on $\widehat\Sigma$. Relying on the conventional theory of Baker--Akhieser functions \cite{Krichever_Itogi}, we will consider as proven the existence of a $d$-dimensional space ${\frak B}(D_p,{\mathbf p})$ of scalar-valued Baker--Akhieser functions on $\widehat \Sigma$ with an arbitrary ${\mathbf p}$ as a set of essential singularities, and a pole divisor $D_p$, where for a $(Q,p_Q)\in{\mathbf p}$ we consider $Q$ as a point of essential singularity, and $p_Q$ as a main part of the exponent at $Q$.

For an arbitrary time $t_a$ we can take $Q\in\pi^{-1}(P_a)$, and $p_Q(w_a)=w_a^{-m_a}\nabla\chi_a(\Lambda(w_a))$, and apply the above statement to the set ${\mathbf p}$ obtained this way. Then functions in the space ${\frak B}(D_p,{\mathbf p})$ have the exponents as claimed in (BA2).

However, we claim existence and uniqueness of the Baker--Akhieser function with an additional requirement (BA4). We will show that the dimension of the space of  modified this way Baker--Akhieser functions does not change.

Let $V\simeq\C^d$ and $V_i=\{ v\in V \ |\ h_\ga v=iv \}$. Obviously,
\begin{equation}\label{E:vgr}
    V=\bigoplus_{i=-m}^m V_i.
\end{equation}

Define a flag $F\, :\, \{ 0\}\subseteq F_{-m}\subseteq\ldots\subseteq F_m=V$ by setting
\begin{equation}\label{E:hflag}
   F_j=\bigoplus_{i=-m}^j V_i.
\end{equation}
Then $F_m=V$. We set $F_i=V$ also for $i>m$. It has been shown in course of the proof of \cite[Lemma 3.4]{Sh_MMO} that for a (locally) holomorphic vector function $\psi_0$ the following expansion takes place:
\[
   z^{h_\ga}\psi_0=\sum_{i=-m}^\infty \psi_0^i z^i
\]
where $\psi_0^i\in F_i$ for all $i\ge -m$. On the one hand, by Riemann--Roch theorem, the contribution of the pole at any $\ga$ to the  dimension of such functions should be equal to $m\dim V$. On the other hand, the total codimension of the relations $\psi_0^i\in F_i$ is also equal to $m\dim V$ which easy follows from the equality $\dim V_i=\dim V_{-i}$. Therefore, the contribution is actually equal to zero. Applying this argument to any column of the matrix-function $z^{h_\ga}\Psi_0$ we obtain that the contribution of this matrix-function into the dimension of the space of Baker--Akhieser functions is equal to $0$. Therefore adding the condition (BA4) leaves the conventional dimension of the space of Baker--Akhieser functions invariant.

By the above procedure of generation of vector-valued and matrix-valued Baker--Akhieser functions, given an arbitrary base $\psi_1,\ldots,\psi_d$ in ${\frak B}(D_p,{\mathbf a})$ we can form a vector-function $\psi(Q)=(\psi_1(Q),\ldots,\psi_d(Q))^T$ on $\widehat\Sigma$. The upper $T$ denotes transposition here. Further on, given a point $q\in\Sigma$, we consider the matrix $\widehat\Psi(q)$ formed by the vector-functions $\psi(Q_j)$ as columns where $j=1,\ldots,d$, $Q_j=(q,\l_j)$, $\l_1,\ldots,\l_d$ are the roots of the spectral curve equation over $q$. $\widehat\Psi$ depends on the order of roots of the spectral equation, but apart that, it is unique up to an appropriate normalization.

It is our next modification that we take $\widehat\Psi(q)\in G$ for all $q\in\Sigma$. With this aim, we first define a bilinear form on the space
${\frak B}(\widehat D,{\mathbf p})$. For $\g=\so(2n),\so(2n+1)$, and  $\psi_1,\psi_2\in{\frak B}(D_p,{\mathbf p})$, we set
\begin{equation}\label{E:scalar_pr}
  (\psi_1,\psi_2)(q)= \frac{1}{2}\sum_{Q:\,\pi(Q)=q}\psi_1(Q)\psi_2(Q^\s)
\end{equation}
where $\pi\, :\widehat\Sigma\to\Sigma$ is the covering map, $q\in\Sigma$, $Q\in\widehat\Sigma$, and $\s\, :\, (q,\l)\to(q,-\l)$ is the involution (where $\l$, $-\l$ is a pair of opposite roots of the characteristic equation $R(\l,q)=0$). Let $P$ be an exponential point of $\widehat\Psi$ (see (BA2)), and $Q$ is in a small enough neighborhood of $P$. Then $\Lambda(Q)=-\Lambda(Q^\s)$ by definition of $\s$. Since $\chi$ is $\s$-invariant, $\nabla\chi$ is $\s$-antiinvariant, i.e. $\nabla\chi(\Lambda(Q))=-\nabla\chi(\Lambda(Q^\s))$. Hence every singularity of exponential type turns to its inverse under the involution $\s$. The same holds for the Tyurin-type poles, because the sum of every two diagonal elements of $h_\ga$ which are permuted by $\s$, vanishes, by definition of a Cartan subalgebra. Hence the just defined point-wise scalar product of Baker--Akhieser vector-functions is a meromorphic function (holomorphic at the just listed points). Choose the above base $\psi_1(Q),\ldots,\psi_d(Q)$ to be orthogonal with respect to the scalar product \refE{scalar_pr}. Then the columns of the corresponding matrix $\widehat\Psi$ will be orthogonal with respect to the scalar product $(x,y)=x\s y^T$ where $x,y\in\C^d$ are considered as rows, $T$ denotes transposition, and, in abuse of notation, $\s$ here denotes an operator in $\C^d$ permuting the standard basis $e_1,\ldots,e_d\in\C^d$ according to the involution $\s$ on the sheets of the covering $\widehat\Sigma\to\Sigma$, namely $e_1,\ldots,e_d \to e_{\s(1)},\ldots,e_{\s(d)}$. Observe that $L$ is considered to be skew-symmetric matrix exactly with respect to this scalar product. Finally, we normalize $\psi_1,\ldots,\psi_d$ so that the matrix $\widehat\Psi\in SO(d)$.

The argument is the same in the symplectic case except that the scalar product \refE{scalar_pr} is replaced with the symplectic form  :
\begin{equation}\label{E:symp_f}
  \langle\psi_1,\psi_2\rangle(q)= \frac{1}{2}\sum_{Q:\,\pi(Q)=q}(\psi_1(Q)\psi_2(Q^\s)-\psi_1(Q^\s)\psi_2(Q)).
\end{equation}
\end{proof}
\begin{theorem}\label{T:invsptr}
Let $ \widehat\Psi$ be given by \refT{BA_exist}, and $h_{\ga(t)}=g(t,z)hg(t,z)^{-1}$ where $h\in\h$ is integral, $g$, $g^{-1}$ and $\dot g$ are holomorphic, $g(t)\in G$ for all $t$. Then
\begin{itemize}
\item[$1^\circ$.]
$L$ and $M_a$ given by \refE{LM} with $t=t_a$ provide a solution of the Lax equation with the relation \refE{sravn} between $L$ and $M_a$.
\item[$2^\circ$.]
$L\in\L_{\Pi,\Gamma,h}^D$, and $M_a\in\M_{\Pi,\Gamma,h}^D$, as defined in  \refSS{Lax_eq}, in particular $\left( L\big|_{\Sigma_D}\right) +D\ge 0$.
\end{itemize}
\end{theorem}
\begin{proof}
By (BA2) $\Lambda$ is a $\h$-valued function where $\h$ is a diagonal (Cartan) subalgebra in~$\g$.  Since by (BA1) $\widehat\Psi(q)\in G\cdot\frak D$, and $\frak D$ commutes with $\h$, $L=\widehat\Psi\Lambda{\widehat\Psi}^{-1}$ is a $\g$-valued function on $\Sigma$. The relation $M=-\partial_t\widehat\Psi\cdot{\widehat\Psi}^{-1}$ obviously implies the same for~$M$. By $t$ we denote any of the times $t_a$.

By differentiating the both sides of the relation $L=\widehat\Psi\Lambda{\widehat\Psi}^{-1}$, taking account of time independence of $\Lambda$, we obtain
\[
   \dot L=\partial\widehat\Psi\cdot\Lambda{\widehat\Psi}^{-1}- \widehat\Psi\Lambda{\widehat\Psi^{-1}}\partial\widehat\Psi\cdot{\widehat\Psi}^{-1}= [\Psi\Lambda\Psi^{-1},-\partial\Psi\cdot\Psi^{-1}]=[L,M].
\]

At a Turin point $z_\ga$, by \refE{TPform} and $L=\widehat\Psi\Lambda\widehat\Psi^{-1}$, we have
\[
   L=(z-z_\ga)^{h_\ga}(\Psi_0\Lambda\Psi_0^{-1})(z-z_\ga)^{-h_\ga}
\]
where $\Psi_0\Lambda\Psi_0^{-1}$ is holomorphic. By \refL{lform}, and for the reason that $\Psi_0\Lambda\Psi_0^{-1}$ is holomorphic, $L$ has a Laurent expansion \refE{ga_expan} at $z=z_\ga$.

Since the exponential factor in (BA2) commutes with $\Lambda$, we have $L=\Psi_P\Lambda\Psi_P^{-1}$ on $\Sigma_D$ where $\Psi_P$ is holomorphic and holomorphically invertible. Hence, taking account of \refE{lower_b} we obtain
\[
 \left(L\big|_{\Sigma_D}\right)=\left(\Lambda\big|_{\Sigma_D}\right)\ge -D.
\]
Next, consider the expression for the $M$-operator at a Tyurin point $\ga$. By differentiation of both parts of (BA4) we obtain
\[
   \partial_t\widehat\Psi\cdot{\widehat\Psi}^{-1}= \partial_t(z-z_\ga)^{h_\ga}\cdot (z-z_\ga)^{-h_\ga}+(z-z_\ga)^{h_\ga}(\partial_t\Psi_0\cdot\Psi_0^{-1})(z-z_\ga)^{-h_\ga}.
\]
To calculate $\partial_t(z-z_\ga)^{h_\ga}\cdot (z-z_\ga)^{-h_\ga}$ we plug $(z-z_\ga)^{h_\ga}=g(z-z_\ga)^hg^{-1}$ where $h$ is time-independent. Then
\[
  \partial_t(z-z_\ga)^{h_\ga}\cdot (z-z_\ga)^{-h_\ga}={\dot g}g^{-1}+ \frac{\nu h_\ga}{z-z_\ga} - (z-z_\ga)^{h_\ga}{\dot g}g^{-1}(z-z_\ga)^{-h_\ga}
\]
where $\nu=\partial_tz_\ga$. We obtain
\begin{equation}\label{E:M-type}
  M={\dot g}g^{-1}+\frac{\nu h_\ga}{z-z_\ga}-(z-z_\ga)^{h_\ga}({\dot g}g^{-1}+\partial_t\Psi_0\cdot \Psi_0^{-1})(z-z_\ga)^{-h_\ga}.
\end{equation}
By assumption, ${\dot g}g^{-1}+\partial_t\Psi_0\cdot \Psi_0^{-1}$ is holomorphic. Therefore by \refL{lform}, the third summand in \refE{M-type} has an $L$-type Laurent expansion. The first occurence of ${\dot g}g^{-1}$ in \refE{M-type} we replace with $O(1)$. We conclude that the the right hand side of \refE{M-type} has an expansion of the type \refE{M_oper}.

\item[$3^\circ$.]
In the neighborhood of an exponential point $P$, plugging (BA2) to the relation $M_a=-\partial_{t_a}\widehat\Psi\cdot{\widehat\Psi}^{-1}$ we obtain
\[
   M_a=-\partial_t{\Psi_P}\cdot{\Psi_P}^{-1}+{\Psi_P}(w^{-m}\nabla\chi(\Lambda)+ O(1)){\Psi_P}^{-1}
\]
Since $\Psi_P$ is holomorphic at $P$, we have
\[
   M_a={\Psi_P}(w^{-m}\nabla\chi(\Lambda)){\Psi_P}^{-1}+ O(1).
\]
By $\Ad$-invariance of the polynomial $\chi$ it follows that ${\Psi}(w^{-m}\nabla\chi(\Lambda)){\Psi}^{-1}=w^{-m}\nabla\chi(L)$ which implies the relation \refE{sravn} between $L$ and $M_a$.
\end{proof}
\subsection{Theta-functional formula for the Baker--Akhieser matrix-function}\label{SS:BAexpl}
The aim of the section is to appropriately modify a conventional Krichever-type theta-functional formula for the Baker--Akhieser function (for a Bloch functions such kind of relation has been introduced by A.Its). By \refE{TPform}, it follows that at any Tyurin point $\ga$, all entries of the $i$th row of the matrix $\widehat\Psi$ have the same order $z^{h_\ga^i}$ where $h_\ga^i$ is the $i$th diagonal element of $h_\ga$. We assign these orders to the $i$th basis element of ${\frak B}(D_p,{\mathbf p})$ by modifying the 1-form involved in such kind of formulas. The last  will be no longer the same for all rows.

For any $P\in{\rm supp}\, D$ let $\{ P_j\ |\ j=1,\ldots,d\}$ be its full preimage in $\widehat\Sigma$. For any invariant polynomial $\chi_a$ let $\mu_a^j(w_{P_j})$ denote     the $j$th diagonal element of $\nabla\chi_a(\Lambda(w_P))$ (which is nothing but the $j$th eigenvalue of $M_a$), where $w_{P_j}$ is a local coordinate in the neighborhood of $P_j$ ($w_{P_j}=w_P$ if $P_j$ is not a branch point of the spectral curve, i.e. in a general position).
\begin{theorem}\label{T:exfor}
Let $\Omega_i$ be the unique 1-form on $\Sigma$ such that
\begin {itemize}
\item[$1^\circ$.] In a neighborhood of any $P_j$
\[
    \Omega_i(w_{P_j})=\left(-\sum_{\substack{a:\, P_a=P,\\ m_a\le m_P} }
    t_a w^{-m_a}_{P_j}\mu_a^j(w_{P_j})+O(1)\right)dw_{P_j};
\]
\item[$2^\circ$.]
In a neighborhood of any Tyrin point $\ga$
\[
    \Omega_i(z)= \left(\frac{h_\ga^i}{z-z_\ga}+O(1)\right)dz.
\]

\item[$3^\circ$.]
$\Omega_i$ is holomorphic outside the above points, and its $a$-periods are equal to zero.
\end{itemize}
Then the Baker--Akhieser matrix
\begin{equation}\label{E:BA_semi}
   {\widehat\Psi}_{ij}(q)=\left(\exp\int\limits^{(q,\l_j)}\Omega_i\right)\frac{\theta({\rm A}(Q)+Z(D_i)+U)}{\theta({\rm A}(Q)+Z(D_i))},\ i=0,\ldots,n-1,
\end{equation}
subjected to the further orthogonalization (skew-orthogonalization, resp.) procedure introduced above, satisfies the conditions of \refT{BA_exist}.
\end{theorem}
\begin{proof}
It is conventional that the $\g$-valued functions given by the relations \refE{LM}, \refE{BA_semi} are well-defined on $\Sigma$ whatever be the differential $\Omega$ satisfying the conditions $1^\circ$, $3^\circ$. It can be proven exactly as in the case $G=GL(n)$ \cite{Krichever_Itogi}.

In presence of involution, since the divisor $D$ is invariant with respect to it, the procedure of orthogonalization (skew-orthogonalization, resp.) can be applied. The problem is to fulfill the conditions (BA2) and, especially, (BA4) in the process. For instance, the condition (BA4) means that the $i$th row of $\widehat\Psi$ has the order $h_\ga^i$ at $\ga$ where $h_\ga=\{ h_\ga^1,\ldots,h_\ga^d \}$. We will show that these orders are preserved in the process of orthogonalization (skew-orthogonalization, resp.), hence (BA4) stays to be fulfilled.

Indeed, let vector-functions $e_1,\ldots,e_k,e_{k+1}$ have the orders $h_1,\ldots,h_k,h_{k+1}$ at a point, say, $z=0$. Thus, locally, $e_i\sim z^{h_i}$, $i=1,\ldots, k+1$. On the $(k+1)$th step of the orthogonalization process we are looking for $\l_1,\ldots,\l_k$ satisfying the equations
\[
     (\l_1e_1+\ldots+\l_ke_k+e_{k+1},e_j)=0,\ j=1,\ldots, k+1.
\]
Then $\l_s=(-1)^{k-s}\det[(e_i,e_j)]_{i=1,\ldots,k}^{j=1,\ldots,{\widehat s},\ldots,k,k+1}/\det[(e_i,e_j)]_{i=1,\ldots,k}^{j=1,\ldots,k}$. Obviously,
\[
\det[(e_i,e_j)]_{i=1,\ldots,k}^{j=1,\ldots,k}\sim z^{2(h_1+\ldots+h_k)},\ \det[(e_i,e_j)]_{i=1,\ldots,k}^{j=1,\ldots,{\widehat s},\ldots,k,k+1}\sim z^{2(h_1+\ldots+h_k)-h_s+h_{k+1}},
\]
hence $\l_s\sim z^{h_{k+1}-h_s}$. Since $e_s\sim z^{h_s}$, all terms of the linear combination $\l_1e_1+\ldots+\l_ke_k+e_{k+1}$ are of the same order $z^{h_{k+1}}$.

As for (BA2), it will be preserved for the reason that the columns stay untouched in the process of orthogonalization of rows.

The resulting Baker--Akhieser matrix satisfies the conditions of \refT{invsptr}.

\end{proof}

\subsection{Explicit expression for the exponents}
Here we would like to give an explicit expression (the relation \refE{eigM1} below) for $\mu_a^j(w_{P_j})$ in \refT{exfor}.

Recall that $\mu_a^j(w_{P_j})$ is the $j$th diagonal element of $\nabla\chi_a(\Lambda(w_P))$, as introduced  in \refSS{BAexpl}. It gives the main part of the differential $\Omega_i$ at the point $P_j$, as formulated in \refT{exfor},$1^\circ$. Thus, it will be our first step to give an explicit expression for $\nabla\chi(\Lambda)$ where $\chi$ is a basis invariant of the Lie algebra $\g$.

\begin{example}
Let $\chi(L)=\tr\, L^{n+1}$. Then $d\chi(L)=\tr\, dL^{n+1}=(n+1)\tr\, L^ndL$, hence $\nabla\chi(L)=(n+1)L^n$.
\end{example}
\noindent This example reproduces the expression for the main part of the $M$-operator in~\cite{Klax}.
\begin{example}\label{Ex:det}
Let $\chi(L)=\det L$. Then $(\nabla\chi)(L)=(\det L)L^{-1}$.
\begin{proof}[Proof 1]
By differentiating determinant as a polylinear functional of its columns, we obtain: $d\det L=\sum_{ij}A_{ij}dL_{ij}$ where $A_{ij}$ is a cofactor of the  $(i,j)$th element of the matrix $L$. Let $A=(A_{ij})$. Then the previous relation can be written down as $d\det L=\tr(A^TdL)$ ($A^T$ being the transposed matrix for $A$). It is conventional that $A^T=(\det L)L^{-1}$, which implies the statement.
\end{proof}
\begin{proof}[Proof 2]
We start with $\det L=\exp\tr\log L$. Then
\[
 d\det L=(\exp\tr\log L)d(\tr\log L)=(\det L)\tr\, d\log L
 =(\det L)\tr\,(L^{-1}dL).
\]
Finally, $d\det L=\tr\,\left((\det L)L^{-1}\cdot dL\right)$.
\end{proof}
\end{example}
\refEx{det} possesses a certain universality, namely it suggests the way to calculate gradients of \emph{all} coefficients of the characteristic polynomial. To be more precise, we will calculate their generating function. On the one hand, such generating function is given by the gradient of the characteristic polynomial itself:
\[
  \nabla\det(\l-L)=\sum_{i=1}^l \l^{n-i}\nabla r_i(L).
\]
On the other hand
\[
  \nabla\det(\l-L)=\det(\l-L)(\l-L)^{-1}.
\]
Consider the characteristic equation in a most general form
\[
        R(q,\l)=\l^{d}+\sum_{j=1}^n \l^{d-d_j}r_j(q)=0
\]
where $d_j$ is a degree of $r_j$ as of a polynomial in $L$.

Assuming $r_0\equiv 1$, $\det L\ne 0$, and taking $L^{-1}$ off, we obtain
\begin{align*}
  \det(\l-L)(\l-L)^{-1}&=-\left( \sum_{j=0}^n \l^{n-d_j}r_j(L)  \right) (1-\l L^{-1})^{-1}L^{-1}\\
  &=-\left( \sum_{j=0}^n \l^{n-d_j}r_j(L)  \right)\left(\sum_{p=0}^\infty (-1)^p{\l^p}L^{-p}\right)L^{-1}.
\end{align*}
Finally, we obtain
\[
  \det(\l-L)(\l-L)^{-1}= \sum_{j=0}^n\sum_{p=0}^\infty (-1)^{p+1}\l^{n-d_j+p}r_j(L)L^{-p-1},
\]
which implies
\[
  \nabla r_i(L)=\sum_{\substack{d_j-p=d_i\\
j\le n,\ p\ge 0 }} (-1)^{p+1} r_j(L)L^{-p-1}.
\]
For a simple Lie algebra this relation writes as
\[
  \nabla r_i(L)=\sum_{\substack{d_j-p=d_i\\
j\le n,\ p\ge 0 }} (-1)^{p+1} r_j(L)L^{-p-1}.
\]
Here, $p=d_j-d_i$, and $p\ge 0$ is equal to $j\ge i$. Therefore, changing the summation limits, we write down the expression for the gradient as follows:
\begin{equation}\label{E:var}
  \nabla r_i(L)=\sum_{i\le j\le n} (-1)^{d_j-d_i+1} r_j(L)L^{d_i-d_j-1}.
\end{equation}
The last is a relation between diagonal matrix-valued functions on $\Sigma$. For the induced scalar-valued functions on the spectral curve we have
\begin{equation}\label{E:eigM}
      \nabla r_i(q,\l)=\sum_{i\le j\le n} (-1)^{d_j-d_i+1} r_j(q)\l^{d_i-d_j-1}.
\end{equation}
This implies the following expression for $\mu_a^j$:
\begin{equation}\label{E:eigM1}
    \mu_a^j(w_{P_j})=\sum_{i\le j\le n} (-1)^{d_j-d_i+1} r_j(w_{P_j})\l^{d_i-d_j-1}_j.
\end{equation}
In the case $\g=\so(2n)$, the $r_n(L)=\det L$ is not a basis spectral invariant, and what we need is a gradient of the Pfaffian. For the last, from $\nabla {\rm Pf}(L)^2=(\det L)L^{-1}$ we derive
\[
  \nabla {\rm Pf}(L)=\frac{1}{2}{\rm Pf}(L)L^{-1}
\]
which implies
\[
  \mu_a^j(w_{P_j})=\frac{1}{2}{\rm Pf}(L(w_{P_j}))\l_j^{-1}
\]
for the times corresponding to the Hamiltonians related to the Pfaffian.

In the case $\g=\so(2n+1)$ $L$ is always degenerate. However, the above calculation can be carried out after excluding the characteristic root $\l=0$ (dividing the both parts of the spectral equation by $\l$) which means that we consider only the non-trivial irreducible component of the spectral curve.
\begin{remark}
In the case the base curve is hyperelliptic (given in the form $y^2=P_{2g+1}(x)$), and $P_j$ runs over the preimages of $\infty$, the asimptotic behavior of $x$, $y$, $\l$ at every $P_j$, as well as explicit expressions for $r_i(x,y)$ (for all $i$), have been explicitly calculated in \cite{Sh_FAN_2019,Sh_Bor}. This enables one to explicitly express $\Omega_i$ in terms of the basis holomorphic Prym differentials which are also known.
\end{remark}
\bibliographystyle{amsalpha}

\begin{thebibliography}{A}
\define\PL{Phys. Lett. B}
\define\NP{Nucl. Phys. B}
\define\LMP{Lett. Math. Phys. }
\define\JGP{JGP}
\redefine\CMP{Commun. Math. Phys. }
\define\JMP{J.  Math. Phys. }
\define\Izv{Math. USSR Izv. }
\define\FA{Funct. Anal. and Appl.}
\def\Pnas{Proc. Natl. Acad. Sci. USA}
\def\PAMS{Proc. Amer. Math. Soc.}
\def\UMNr{Russ. Math. Surv.}

\bibitem{Sh_Bor}
Borisova, P.I., Sheinman, O.K. \emph{Hitchin systems on hyperelliptic curves}. Proceedings of Steklov mathematical institute, Vol. 311 (2020) (to be printed); 	 arXiv:1912.06849 [math-ph].

\bibitem{Kr_UMN1977}
I. M. Krichever. \emph{Methods of algebraic geometry in the theory of non-linear equations}. Russian Math. Surveys, 32:6 (1977), 185-213.

\bibitem{Krichever_Itogi}
Krichever, I.M. \emph{Nonlinear equations and elliptic curves.} J. Math. Sci. 28, 51–90 (1985). https://doi.org/10.1007/BF02104896.

\bibitem{Klax}
Krichever, I.M. \emph{Vector bundles and Lax equations on
algebraic curves}. Comm. Math. Phys. Vol. 229, 229--269 (2002).

\bibitem{KN_UMN1980}
Krichever, I.M., Novikov, S.P. \emph{Holomorphic bundles over algebraic curves and non-linear equations}. Russian Math. Surveys, 35:6 (1980), 53-79.

\bibitem{Kr_Sh_FAN}
Krichever, I.M., Sheinman, O.K. \emph{Lax Operator Algebras}. Funct. Anal. Appl., 41:4 (2007), 284-294; 	arXiv:math/0701648 [math.RT].

\bibitem{Sh_DGr}
O.K. Sheinman.   \emph{Current algebras on Riemann surfaces}. De Gruyter Expositions in Mathematics, 58, Walter de Gruyter GmbH, Berlin--Boston, 2012, ISBN: 978-3-11-026452-4, 150 p.

\bibitem{Shein_UMN2016}
O.K.Sheinman. \emph{Lax operator algebras and integrable systems}, Russian Math. Surveys, 71:1 (2016), 109--156; arXiv:1602.04320

\bibitem{Sh_TG2016}
O.K. Sheinman, “Lax operator algebras and gradings on semisimple Lie algebras”, Transform. Groups, 21:1 (2016), 181--196; 	arXiv:1406.5017 [math.RA].

\bibitem{Sh_FAN_2019}
Sheinman, O.K. \emph{Spectral curves of hyperelliptic Hitchin systems}. Funct. Anal. Appl., 53:4 (2019), 291–303; arXiv:1806.10178 [math-ph].

\bibitem{Sh_MMO}
O.K.Sheinman. \emph{Matrix divisors on Riemann surfaces and Lax operator algebras}. Trans. Moscow Math. Soc., 2017, P. 109–121.
http://dx.doi.org/10.1090/mosc/267; arXiv:1701.01807 [math.AG].

\end{thebibliography}

\end{document}